\documentclass[conference]{IEEEtran}
\IEEEoverridecommandlockouts

\usepackage{amsmath}
\usepackage{amssymb}
\usepackage{amsthm}
\usepackage{graphicx}
\usepackage{epstopdf}
\usepackage{bm}
\usepackage{color}
\usepackage{cite}
\usepackage{subfigure}
\newtheorem{theorem}{Theorem}

\newtheorem{definition}{Definition}

\newtheorem{remark}{Remark}

\begin{document}
\title{A Game-Theoretic Approach to Covert Communications} 

\author{Alex S. Leong, Daniel E. Quevedo,  and Subhrakanti Dey
\thanks{A. Leong and D. Quevedo are with the Department of Electrical Engineering (EIM-E), Paderborn University, Paderborn, Germany.  E-mail: {\tt alex.leong@upb.de, dquevedo@ieee.org.} 
  S. Dey is with the Department of Electronic Engineering, Maynooth University, Maynooth, Ireland. E-mail: {\tt subhra.dey@mu.ie.}, and also with Uppsala University, Sweden. }
}

\maketitle

\begin{abstract}
This paper considers a game-theoretic formulation of the covert communications problem with finite blocklength, where the transmitter (Alice) can randomly vary her transmit power in different blocks, while the warden (Willie) can randomly vary his detection threshold in different blocks. In this two player game, the payoff for Alice is a combination of the coding rate to the receiver (Bob) and the detection error probability at Willie, while the payoff for Willie is the negative of his detection error probability. Nash equilibrium solutions to the game are obtained, and shown to be efficiently computable using linear programming. For less covert requirements, our game-theoretic approach can achieve significantly higher coding rates than uniformly distributed transmit powers. We then consider the situation with an additional jammer, where Alice and the jammer can both  vary their powers. We pose a two player game where Alice and the jammer jointly comprise one player, with Willie the other player.  The use of a jammer is shown in numerical simulations to lead to further significant performance improvements.

\end{abstract}

\section{Introduction}
In covert communications, a transmitter (Alice) transmits to a receiver (Bob) in the presence of a warden (Willie). The aim is for the transmission to be such that the very presence of a transmission or non-transmission is difficult for Willie to distinguish between \cite{BashGoeckelTowsley,Bash_magazine}. Applications of covert communication include the prevention of knowledge of transmission for use as metadata or to maintain privacy, communication in the presence of authoritarian governments, and military communications where detection of transmissions can reveal one's location to enemies \cite{SobersBash}.  

In \cite{BashGoeckelTowsley} it was shown that Alice can transmit $O(\sqrt{N})$ bits in $N$ channel uses covertly and reliably to Bob as $N \rightarrow \infty$. Covertness is defined in the sense that 
\begin{equation}
\label{covertness_defn_asymptotic}
\mathbb{P}_{FA} + \mathbb{P}_M \geq 1 - \epsilon \textnormal{ for any }\epsilon > 0,
\end{equation} 
with $\mathbb{P}_{FA}$ denoting the probability of false alarm and $\mathbb{P}_M$ the probability of missed detection. Further refinements of this result include \cite{CheBakshiJaggi,Bloch_covert,WangWornellZheng}.
Later, it was shown that in certain situations, it is possible to transmit $O(N)$ bits in $N$ channel uses as $N \rightarrow \infty$, such as when there is uncertainty in the receiver noise variance \cite{LeeBaxley}, or when there is an uninformed jammer \cite{SobersBash}. 

The above results are asymptotic in that the results apply for $N \rightarrow \infty$. The case of finite $N$ has been considered in \cite{Yan_delay_intolerant}, where expressions for $\mathbb{P}_{FA}$ and $\mathbb{P}_M$ were derived, and the use of uniformly distributed transmission powers was also proposed as a way to improve performance over the use of constant powers. The current paper also considers the case of finite $N$. Instead of uniformly distributed transmission powers, we instead wish to find the ``optimal'' distribution of transmit powers. Note that if Alice knows the detection threshold that Willie uses, then such an optimal distribution can be found. On the other hand, Willie himself could also try to randomly vary his detection threshold to confuse Alice and potentially improve his detection performance. Due to the competing objectives for Alice and Willie,  in this paper we will use game theory to model such interactions. We will formulate the situation as a two player nonzero sum/zero sum game, and show that Nash equilibrium solutions can be computed efficiently using linear programming. 

We then consider the case where there is also a jammer \cite{SobersBash}, where we now allow both the transmission and jamming powers to randomly vary. Here we formulate a two player game where Alice and the jammer jointly form one player, while Willie is the other player. We similarly show that Nash equilibria can be computed using linear programming. It should be noted that a recent work \cite{covert-games-cicc} has considered a power-threshold game in a {\em non-randomized setting without a jammer}, where Alice and Willie choose their power and  threshold respectively, in a deterministic fashion. Standard Nash equilibrium was derived in this case along with a Bayesian game formulation for the case where Willie's noise power is not known to Alice exactly, but only in distribution.

The paper is organized as follows. The system model is presented in Section \ref{model_sec}. The game-theoretic formulation is presented in Section \ref{game_theoretic_sec}. Section \ref{jammer_power_sec} extends the results to the case with an additional jammer.  Numerical studies and comparisons are given in Section \ref{numerical_sec}.  

\section{System Model}
\label{model_sec}

\begin{figure}[t!]
\centering 
\includegraphics[scale=0.4]{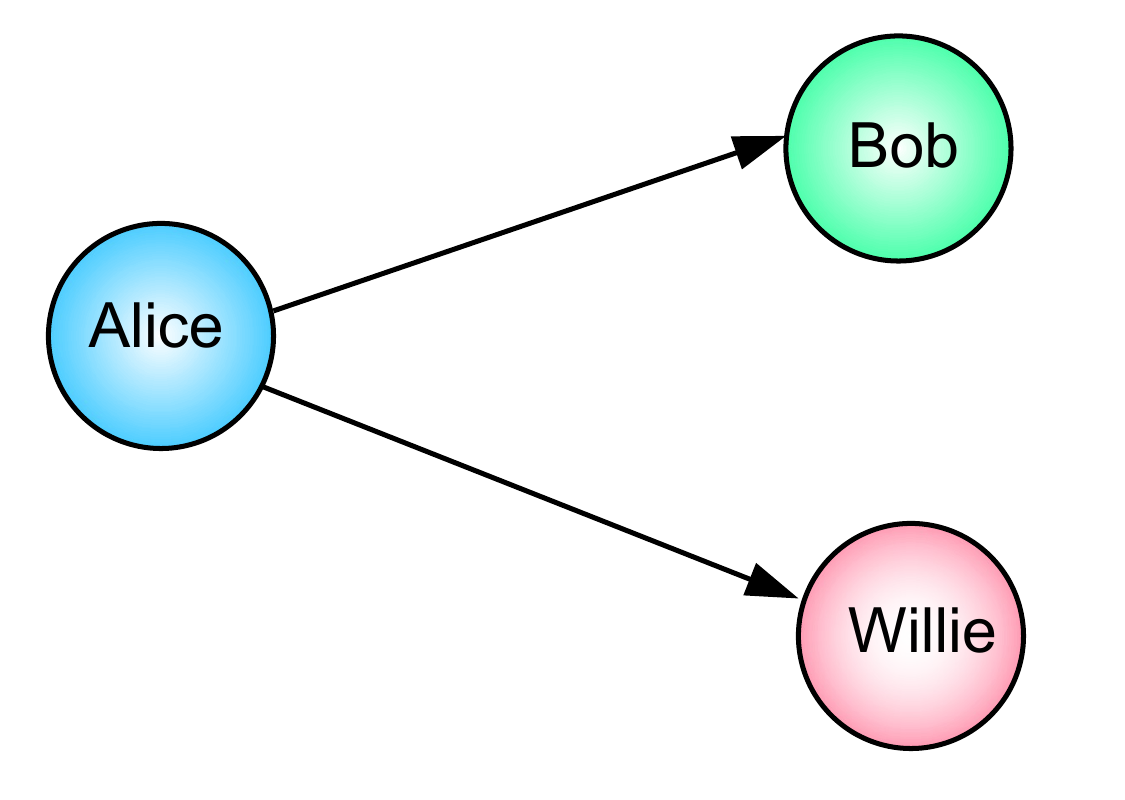} 
\caption{System model - Transmitter varying transmit power}
\label{system_model}
\end{figure} 
A diagram of the system model is shown in Fig. \ref{system_model}.
Let $x_k$ be the signal that is to be transmitted. 
The warden (Willie) wishes to decide between two hypotheses:
\begin{align*}
\mathcal{H}_0:  y_{w,k} &= n_{w,k}, & k = 1,\dots,N\\
\mathcal{H}_1:  y_{w,k} &= x_k + n_{w,k}, & k = 1,\dots,N
\end{align*} 
based on collecting $N$ observations, where $y_{w,k}$ is the received signal by Willie at time $k$, and $n_{w,k} \sim \mathcal{CN}(0,\sigma_w^2)$ is complex Gaussian channel noise. Hypothesis $\mathcal{H}_0$ means that the transmitter (Alice) did not transmit to the receiver (Bob), while hypothesis $\mathcal{H}_1$ means that Alice transmitted.  We assume that the coding blocklength is equal to $N$. The received signals at Bob under the different hypotheses are:
\begin{align*}
\mathcal{H}_0: y_{b,k} &= n_{b,k}, & k = 1,\dots,N\\
\mathcal{H}_1: y_{b,k} &= x_k + n_{b,k}, & k = 1,\dots,N
\end{align*} 
where $n_{b,k} \sim \mathcal{CN}(0,\sigma_b^2)$.

We assume Gaussian signalling such that $x_k \sim \mathcal{CN}(0,P)$. The transmit power $P$ varies between different blocks, but stays constant within each block of $N$ time slots. We assume that Bob knows the (random) values of $P$ used in each block via some shared secret between Alice and Bob, but that Willie only knows the distribution of $P$.\footnote{In game theoretic terminology this is equivalent to saying that Willie knows the mixed strategy that Alice will play.} 

Willie wants to detect transmissions of Alice. Optimal detection at Willie usually takes on the form of a likelihood ratio test \cite{Kay_detection,Poor_book}. Given $\mathcal{H}_0$, we have $y_{w,k} \sim \mathcal{CN} (0,\sigma_w^2)$, and given $\mathcal{H}_1$, we have $y_{w,k} \sim \mathcal{CN} (0,P + \sigma_w^2)$ for $k=1,\dots,N$. Then the likelihood ratio test can be easily shown to be equivalent to an energy detector which decides $\mathcal{H}_1$ if 
\begin{equation}
\label{T_defn}
T \triangleq \frac{1}{N} \sum_{k=1}^N |y_{w,k}|^2
\end{equation}
 exceeds a threshold $t$, and  decides $\mathcal{H}_0$ otherwise \cite{Kay_detection}. 

In covert communications, Alice wants to transmit to Bob while ensuring that the probability of being detected at Willie is sufficiently low \cite{Bash_magazine}.  One strategy for Alice to improve her performance (e.g. in terms of the transmission rate to Bob, or the detection probability at Willie) is by randomizing between a few different transmission powers, with the aim of confusing Willie. 
In \cite{Yan_delay_intolerant} the case of uniformly distributed $P$ was considered and shown to outperform the use of constant $P$. For the current paper we consider the problem of optimizing the distribution for $P$. Suppose that $P>0$ can take on a finite number of values 
$$P_1, P_2, \dots, P_I,$$
and denote
$$\pi_i^P \triangleq \mathbb{P}(P = P_i), \quad i = 1,\dots,I.$$

Now if Willie uses a fixed detection threshold $t$, then Alice can optimize her transmission power distribution for that particular threshold.\footnote{For instance, one can pose a problem of maximizing the transmission rate to Bob while constraining the detection error probability for Willie.} However, if Willie decides to randomize his detection threshold, he in turn could confuse Alice and possibly increase his detection performance. Due to the competing objectives of Alice and Willie, in this paper we will adopt a game-theoretic formulation of the situation, which will be presented in Section \ref{game_theoretic_sec}. 
We thus assume that $t$ can take on values
$$t_1, t_2, \dots, t_M$$
with 
$$\pi_m^t \triangleq \mathbb{P}(t = t_m), \quad m = 1, \dots, M.$$
The case where $t$ can take on a continuum of values can be approximated by discretization of the real interval using a large number of discretization points.

The statistic $T$ defined in (\ref{T_defn}) is equivalent to a scaled chi-squared distributed random variable with $2N$ degrees of freedom under both hypotheses, with scaling $\frac{\sigma_w^2}{2N}$ under $\mathcal{H}_0$, and scaling $\frac{P+\sigma_w^2}{2N}$ under $\mathcal{H}_1$ and transmit power $P$.  The likelihood functions of $T$ under $\mathcal{H}_0$ and $\mathcal{H}_1$ are then
\begin{align*}
f(T|\mathcal{H}_0) & = \frac{T^{N-1}}{\Gamma(N)} \left(\frac{N}{\sigma_w^2} \right)^N \exp\left(- \frac{NT}{\sigma_w^2}\right) \\
f(T|\mathcal{H}_1) & = \frac{T^{N-1}}{\Gamma(N)} \sum_{i=1}^I \left(\frac{N}{P_i+\sigma_w^2} \right)^N \exp\left(- \frac{NT}{P_i+\sigma_w^2}\right) \pi_i^P
\end{align*}
where $\Gamma(.)$ is the gamma function. 
Let $\mathbb{P}_{FA} = \mathbb{P} (\textnormal{decide } \mathcal{H}_1 | \mathcal{H}_0)$ and $\mathbb{P}_M = \mathbb{P} (\textnormal{decide } \mathcal{H}_0 | \mathcal{H}_1)$ denote the probability of false alarm and probability of missed detection respectively. We will say that the  communication scheme is covert  \cite{Yan_delay_intolerant} if\footnote{As we are considering finite blocklengths, we do not consider arbitrarily small $\epsilon$ in the sense of (\ref{covertness_defn_asymptotic}).}  
$$\mathbb{P}_{FA} + \mathbb{P}_M \geq 1 - \epsilon \textnormal{ for some } \epsilon > 0.$$

From the relation
\begin{equation}
\label{integral_relation}
\int T^{N-1} \exp\left( - \frac{NT}{x}\right) dT= - \left( \frac{N}{x}\right)^{-N} \Gamma\Big(N,\frac{NT}{x}\Big) 
\end{equation}
where $$\Gamma(s,x) = \int_x^\infty t^{s-1} e^{-t} dt$$ is the incomplete gamma function, one can easily show that for given distributions of transmit powers $\pi^P \triangleq (\pi^P_1,\dots,\pi^P_I)$ and detection thresholds $\pi^t \triangleq (\pi^t_1,\dots,\pi^t_M)$, the probabilities of false alarm and missed detection are 
\begin{equation*}
\begin{split}
\mathbb{P}_{FA} (\pi^P,\pi^t) & = \mathbb{P} (T > t |\pi^P,\pi^t, \mathcal{H}_0) = \sum_{m=1}^M \frac{\Gamma(N,\frac{N t_m}{\sigma_w^2})}{\Gamma(N)} \pi_m^t \\
\mathbb{P}_{M}  (\pi^P,\pi^t) & =\mathbb{P} (T < t | \pi^P,\pi^t, \mathcal{H}_1) \\ &= \sum_{m=1}^M \sum_{i=1}^I \left[1 - \frac{\Gamma(N,\frac{N t_m}{P_i+\sigma_w^2})}{\Gamma(N)} \right] \pi_i^P \pi_m^t.
\end{split}
\end{equation*}
Note that the expression for $\mathbb{P}_{FA}(\pi^P,\pi^t)$ does not actually depend on $\pi^P$, but for notational consistency with Section \ref{jammer_power_sec} we will use  $\mathbb{P}_{FA}(\pi^P,\pi^t)$ rather than  $\mathbb{P}_{FA}(\pi^t)$.

\section{Game-theoretic formulation}
\label{game_theoretic_sec}
For finite blocklengths, the channel coding rate from Alice to Bob in bits per channel use is approximated by (see \cite{Gursoy_ICC,PolyanskiyPoorVerdu})
\begin{equation}
\label{finite_blocklength_rate_approx}
 R \approx \log_2(1+\textnormal{SNR}_b) - \sqrt{\frac{1}{N} \left(1 - \frac{1}{(\textnormal{SNR}_b+1)^2} \right)} \frac{Q^{-1}(\delta)}{\ln(2)}
 \end{equation}
where $\textnormal{SNR}_b  $ is the signal-to-noise ratio at Bob, $Q^{-1}(.)$ is the inverse $Q$-function, and $\delta$ is the decoding error probability. For future reference, define the function $\overline{R}(.)$ by
\begin{equation}
\label{R_bar_defn}
\overline{R}(x) \triangleq \log_2(1+x) - \sqrt{\frac{1}{N} \left(1 - \frac{1}{(x+1)^2} \right)} \frac{Q^{-1}(\delta)}{\ln(2)}
\end{equation}

In this section we consider posing the situation in Section \ref{model_sec} as a two player game between Alice and Willie, where we wish to find Nash equilibrium solutions to the game. It is well known that for finite games, mixed strategy Nash equilibria always exist.
Here the mixed strategies for  Alice and Willie are $\pi^P$ and $\pi^t$ respectively. 

For transmit power $P$, the signal-to-noise ratio at Bob is $\textnormal{SNR}_b = \frac{P}{\sigma_b^2} $. Alice wants to maximize the payoff
\begin{equation}
\label{payoff_Alice}
\sum_{i=1}^I \overline{R}\left(\frac{P_i}{\sigma_b^2}\right) \pi_i^P + \beta (\mathbb{P}_{FA} (\pi^P,\pi^t)+\mathbb{P}_{M} (\pi^P,\pi^t))
\end{equation}
where $\overline{R}(.)$ is defined by (\ref{R_bar_defn}) and the parameter $\beta> 0$ controls the tradeoff between the (approximate) expected channel coding rate at Bob and covertness at Willie. Smaller values of $\beta$ will place more emphasis on achieving a large coding rate, while larger values of $\beta$ will have more emphasis on achieving higher detection error probabilities (i.e. be more covert).
 Willie on the other hand wants to minimize $\mathbb{P}_{FA} (\pi^P,\pi^t)+\mathbb{P}_{M} (\pi^P,\pi^t)$, so he has payoff
\begin{equation}
\label{payoff_Willie_non_zero_sum}
 - (\mathbb{P}_{FA} (\pi^P,\pi^t)+\mathbb{P}_{M} (\pi^P,\pi^t)).
 \end{equation}
This game with payoffs (\ref{payoff_Alice}) and (\ref{payoff_Willie_non_zero_sum}) for Alice and Bob respectively  is a non-zero-sum game. Nash equilibria to general non-zero-sum games can be found numerically using algorithms such as the Lemke-Howson algorithm \cite{LemkeHowson,NisanRoughgarden}. 

An alternative zero-sum game can also be posed, where Alice has payoff (\ref{payoff_Alice})
and Willie has payoff 
\begin{equation}
\label{payoff_Willie_zero_sum}
- \sum_{i=1}^I \overline{R}\left(\frac{P_i}{\sigma_b^2}\right) \pi_i^P - \beta (\mathbb{P}_{FA} (\pi^P,\pi^t)+\mathbb{P}_{M} (\pi^P,\pi^t)).
\end{equation}
The payoff for Willie can be motivated by saying that in addition to wanting to minimize $\mathbb{P}_{FA} (\pi^P,\pi^t)+\mathbb{P}_{M} (\pi^P,\pi^t)$, Willie also prefers Alice to achieve a lower rate. However, it turns out that the Nash equilibria for both the non-zero-sum and zero-sum games are the same. 

\begin{theorem}
\label{non_zero_sum_lemma}
The non-zero-sum game with payoffs (\ref{payoff_Alice}) and (\ref{payoff_Willie_non_zero_sum}), and the zero-sum game with payoffs (\ref{payoff_Alice}) and (\ref{payoff_Willie_zero_sum}), have the same Nash equilibria. 
\end{theorem}

\begin{proof}
Let $(\pi^{P*},\pi^{t*})$ be a Nash equilibrium to the non-zero-sum game with payoffs (\ref{payoff_Alice}) and (\ref{payoff_Willie_non_zero_sum}). For fixed $\pi^{t*}$, as the payoff (\ref{payoff_Alice}) for Alice is the same in both games, there is no incentive for Alice to deviate from $\pi^{P*}$ in the zero-sum game. While for fixed  $\pi^{P*}$, as $- \sum_{i=1}^I \overline{R}\left(\frac{P_i}{\sigma_b^2}\right) \pi_i^P$ does not depend on $\pi^t$, optimizing (\ref{payoff_Willie_zero_sum}) over $\pi^t$ is equivalent to optimizing (\ref{payoff_Willie_non_zero_sum}), and thus there is no incentive for Willie to deviate from $\pi^{t*}$ in the zero-sum game. Hence $(\pi^{P*},\pi^{t*})$  is also a Nash equilibrium to the zero-sum game with payoffs (\ref{payoff_Alice}) and (\ref{payoff_Willie_zero_sum}). 

A similar argument can be used to show that Nash equilibria to the zero-sum game are also Nash equilibria to the non-zero-sum game. 
\end{proof}

One of the advantages of zero-sum games is that they can be solved efficiently using linear programming \cite{NisanRoughgarden} (note that the Lemke-Howson algorithm itself is similar to the simplex algorithm). A Nash equilibrium  mixed strategy for Alice can be found by solving the linear program:
\begin{align}
\max_{\{\pi_i^P\}, U} \,&  U \nonumber \\
\textnormal{s.t. } & \sum_{i=1}^I \bigg[\overline{R}\bigg(\frac{P_i}{\sigma_b^2}\bigg) \nonumber \\ & + \beta\bigg(\frac{\Gamma(N,\frac{N t_m}{\sigma_w^2})}{\Gamma(N)}  + 1 - \frac{\Gamma(N,\frac{N t_m}{P_i+\sigma_w^2})}{\Gamma(N)}\bigg) \bigg] \pi_i^P \geq U, \nonumber \\ & \quad m = 1,\dots,M, \nonumber \\
& \sum_{i=1}^I \pi_i^P = 1, \quad 0 \leq \pi_i^P \leq 1, \label{lin_prog_Alice}
\end{align}
while a Nash equilibrium  mixed strategy for Willie can be found by solving the linear program:
\begin{align}
\min_{\{\pi_m^t\}, U} \,&  U \nonumber \\
\textnormal{s.t. } & \sum_{m=1}^M \bigg[ \overline{R}\bigg(\frac{P_i}{\sigma_b^2}\bigg) \nonumber \\ &  + \beta\bigg(\frac{\Gamma(N,\frac{N t_m}{\sigma_w^2})}{\Gamma(N)}  + 1 - \frac{\Gamma(N,\frac{N t_m}{P_i+\sigma_w^2})}{\Gamma(N)}\bigg) \bigg] \pi_m^t \leq U, \nonumber \\ & \quad i = 1,\dots,I, \nonumber \\
& \sum_{m=1}^M \pi_m^t = 1, \quad 0 \leq \pi_m^t \leq 1. \label{lin_prog_Willie}
\end{align}

Another advantage of zero-sum games is that their Nash equilibria  have nice ``uniqueness'' properties. We first give the following definition (see also \cite[p.233]{Binmore_book}): 
\begin{definition}
Two Nash equilibria $(\pi^P,\pi^t)$ and $(\pi^{P'}, \pi^{t'})$ are:
\\(i) interchangeable if $(\pi^P,\pi^{t'})$ and $(\pi^{P'}, \pi^{t})$ are also Nash equilibria 
\\(ii) equivalent if the payoffs from using the mixed strategy $(\pi^P,\pi^t)$ are the same as the payoffs from using the mixed strategy $(\pi^{P'}, \pi^{t'})$.
\end{definition}
The following is a standard result in game theory, see e.g. \cite[p.232]{Binmore_book} for a proof.
\begin{theorem}
\label{zero_sum_game_thm}
All Nash equilibria in zero-sum games are interchangeable and equivalent.
\end{theorem} 

We have shown in Theorem \ref{non_zero_sum_lemma} that our original game with payoffs (\ref{payoff_Alice}) and (\ref{payoff_Willie_non_zero_sum}) has the same Nash equilibria as the zero-sum game with payoffs  (\ref{payoff_Alice}) and (\ref{payoff_Willie_zero_sum}). A Nash equilibrium to this zero-sum game can be found by solving the linear programs (\ref{lin_prog_Alice})-(\ref{lin_prog_Willie}). By Theorem \ref{zero_sum_game_thm}, this Nash equilibrium has performance as good any other  Nash equilibria of the game. 
Hence there is no loss of performance in using the mixed strategies obtained by solving the linear programs (\ref{lin_prog_Alice})-(\ref{lin_prog_Willie}).

\section{Presence of a Cooperative Jammer}
\label{jammer_power_sec}
In this section we extend our setup to the situation where there is also a jammer \cite{SobersBash}, which generates a  jamming signal to enhance covertness. It is known \cite{SobersBash} that by using a jammer with jamming powers unknown to Willie, the transmit powers of Alice do not need to go to zero (as the blocklength increases) in order to remain covert in the sense of \cite{BashGoeckelTowsley}. In this paper, we will consider the scenario where Alice and the jammer cooperate by   optimizing of the joint distribution of transmit and jamming powers.

\subsection{System Model}
\begin{figure}[t!]
\centering 
\includegraphics[scale=0.4]{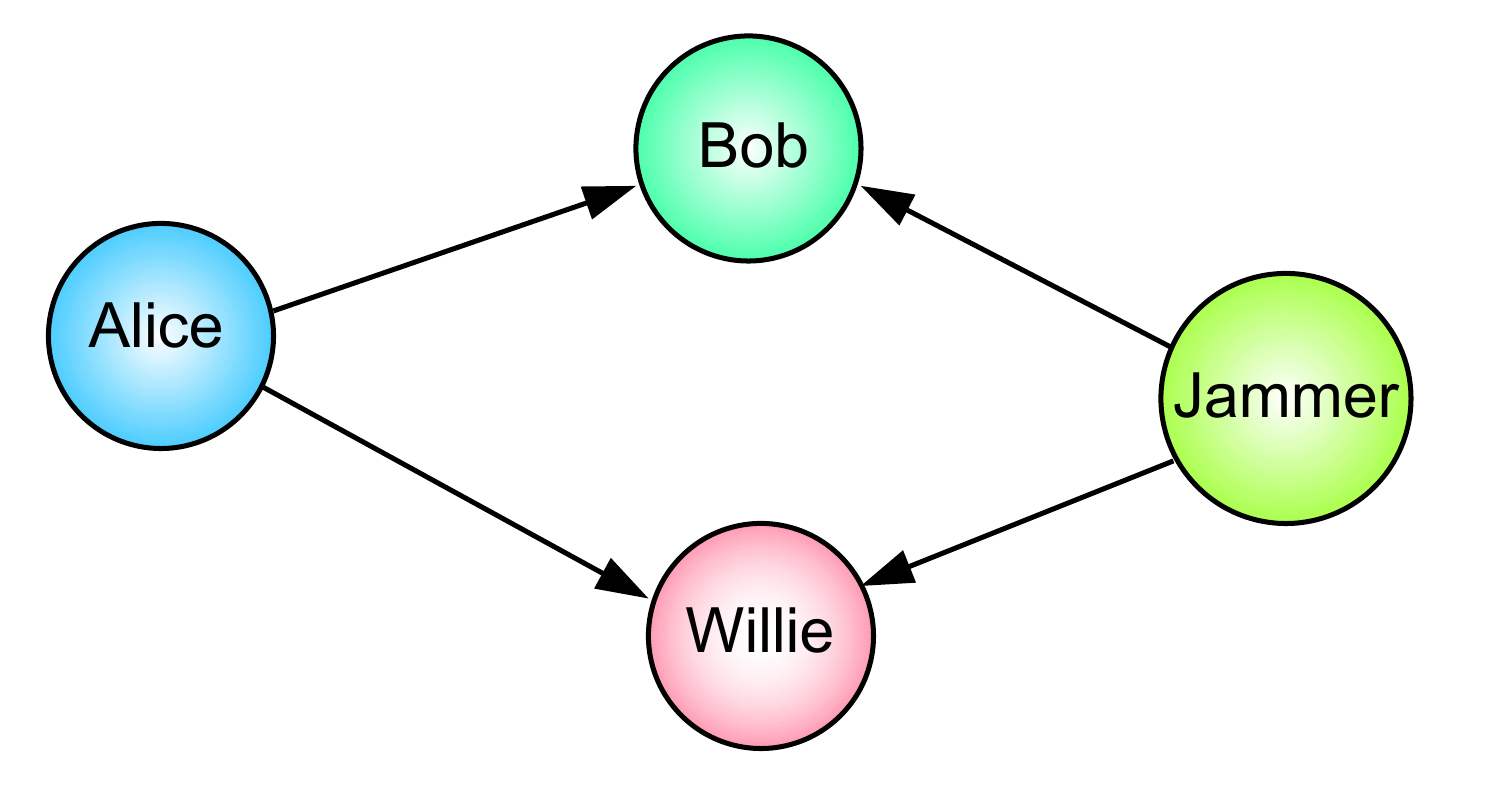} 
\caption{System model - Presence of a jammer}
\label{system_model_jammer}
\end{figure}
A diagram of the system model for this setup is shown in Fig. \ref{system_model_jammer}.
Let $x_k$  again denote the signal which is to be transmitted. 
Willie now wishes to decide between the two hypotheses:
\begin{align*}
\mathcal{H}_0: y_{w,k} &= n_{w,k} + j_k, & k = 1,\dots,N\\
\mathcal{H}_1:  y_{w,k} &= x_k + n_{w,k} + j_k, & k = 1,\dots,N
\end{align*} 
where $j_k \sim \mathcal{CN}(0, J)$ is the random jamming signal. The jamming signal power $J\geq 0$ varies randomly between different blocks, but stays constant within each block of $N$ time slots. 
As before, we assume Gaussian signalling such that $x_k \sim \mathcal{CN}(0,P)$, with $P$ varying randomly between blocks.

The received signals at Bob under the two hypotheses are:
\begin{align*}
\mathcal{H}_0: & \quad y_{b,k} = n_{b,k} + \alpha j_k, & k = 1,\dots,N\\
\mathcal{H}_1: & \quad y_{b,k} = x_k + n_{b,k} + \alpha j_k, & k = 1,\dots,N
\end{align*} 
where $\alpha > 0$ can be used to model different distances between the jammer and Bob, and between the jammer and Willie. We assume a  cooperative jammer such that the transmit powers $P$ and jamming powers $J$ used in each block are known to Bob but unknown to  Willie. The actual values of the random jamming signal $j_k$ are unknown to either Bob or Willie.

We suppose that $P>0$ can take on  values
$$P_1, P_2, \dots, P_I,$$
while $J\geq 0$ can take on values 
$$J_1, J_2, \dots, J_{L}.$$
The joint probabilities of transmit and jamming powers are denoted by
$$\pi_{i,l}^{P,J} \triangleq \mathbb{P}(P=P_i,J = J_l), \quad i = 1,\dots,I, \,l=1,\dots,L.$$
The detection thresholds 
$t$ can take on values
$$t_1, \dots, t_M$$
with 
$$\pi_m^t \triangleq \mathbb{P}(t = t_m), \quad m = 1, \dots, M.$$

The likelihood functions of $T = \frac{1}{N} \sum_{k=1}^N |y_{w,k}|^2$ under $\mathcal{H}_0$ and $\mathcal{H}_1$ are now
\begin{align*}
f(T|\mathcal{H}_0) & = \frac{T^{N-1}}{\Gamma(N)} \sum_{i=1}^{I} \sum_{l=1}^L \left(\frac{N}{\sigma_w^2+J_l} \right)^N \exp\left(- \frac{NT}{\sigma_w^2+J_l}\right) \pi_{i,l}^{P,J}\\
f(T|\mathcal{H}_1) & = \frac{T^{N-1}}{\Gamma(N)} \sum_{i=1}^{I}  \sum_{l=1}^L \left(\frac{N}{P_i+\sigma_w^2+J_l} \right)^N \\ & \quad \quad \times \exp\left(- \frac{NT}{P_i+\sigma_w^2+J_l}\right)  \pi_{i,l}^{P,J}.
\end{align*}

Using again the relation (\ref{integral_relation}), one can now show that for given $\pi^{P,J} \triangleq \{\pi^{P,J}_{i,l}: i=1,\dots,I,l=1,\dots,L\}$ and $\pi^t$,
\begin{equation*}
\begin{split}
\mathbb{P}_{FA} (\pi^{P,J},\pi^t) & = \mathbb{P} (T > t |\pi^{P,J},\pi^t, \mathcal{H}_0) \\ &= \sum_{m=1}^M  \sum_{i=1}^{I} \sum_{l=1}^L  \frac{\Gamma(N,\frac{N t_m}{\sigma_w^2+ J_l})}{\Gamma(N)} \pi_{i,l}^{P,J} \pi_m^t \\
\mathbb{P}_{M}  (\pi^{P,J},\pi^t) & =\mathbb{P} (T < t |\pi^{P,J}, \pi^t, \mathcal{H}_1) \\ &= \sum_{m=1}^M  \sum_{i=1}^{I} \sum_{l=1}^L \left[1 - \frac{\Gamma(N,\frac{N t_m}{P_i+\sigma_w^2+ J_l})}{\Gamma(N)} \right] \pi_{i,l}^{P,J} \pi_m^t.
\end{split}
\end{equation*}

\subsection{Game-Theoretic Formulation} 
Given transmit power $P$ and jamming power $J$, the signal-to-noise ratio at Bob is now
$\textnormal{SNR}_b = \frac{P}{\sigma_b^2 + \alpha^2 J}$.
We will formulate a two player game where the players are 1) Alice-jammer (Alice and the jammer jointly regarded as a single player), and 2) Willie, with mixed strategies $\pi^{P,J}$ and $\pi^t$ respectively. 
Alice-jammer  wants to jointly maximize the payoff
\begin{equation}
\label{payoff_jammer}
\begin{split}
& \sum_{i=1}^{I}  \sum_{l=1}^L  \overline{R} \left( \frac{P_i}{\sigma_b^2 + \alpha^2 J_l} \right) \pi_{i,l}^{P,J} \\ & \quad + \beta' (\mathbb{P}_{FA} (\pi^{P,J},\pi^t)+\mathbb{P}_{M} (\pi^{P,J},\pi^t)),
\end{split}
\end{equation}
where $\overline{R}(.)$ is defined in (\ref{R_bar_defn}) and $\beta' > 0$ controls the tradeoff between the coding rate at Bob and covertness at Willie. 
 Willie on the other hand wants to minimize $\mathbb{P}_{FA} (\pi^{P,J},\pi^t)+\mathbb{P}_{M} (\pi^{P,J},\pi^t)$, so he has payoff
\begin{equation}
\label{payoff_Willie_non_zero_sum_jammer}
 - (\mathbb{P}_{FA} (\pi^{P,J},\pi^t)+\mathbb{P}_{M} (\pi^{P,J},\pi^t)).
 \end{equation}

An alternative zero-sum game can  be posed, where Alice-jammer has payoff 
(\ref{payoff_jammer})
and Willie has a payoff which is the negative of (\ref{payoff_jammer}).

\begin{theorem}
\label{non_zero_sum_lemma2}
The non-zero-sum game with payoffs (\ref{payoff_jammer}) and (\ref{payoff_Willie_non_zero_sum_jammer}), and the zero-sum game with payoffs (\ref{payoff_jammer}) and the negative of  (\ref{payoff_jammer}), have the same Nash equilibria. 
\end{theorem}

\begin{proof}
Similar to the proof of Theorem \ref{non_zero_sum_lemma}.
\end{proof}

A Nash equilibrium mixed strategy for Alice-jammer can be found by solving the linear program:
\begin{align}
\label{Alice_jammer_LP}
&\max_{\{\pi_{i,l}^{P,J}\}, U} \,  U \nonumber \\
&\textnormal{s.t. }  \sum_{i=1}^{I} \sum_{l=1}^L \bigg[ \overline{R} \bigg( \frac{P_i}{\sigma_b^2 + \alpha^2 J_l} \bigg) \nonumber \\ & \quad + \beta'\bigg( \frac{\Gamma(N,\frac{N t_m}{\sigma_w^2+ J_l})}{\Gamma(N)}  + 1 - \frac{\Gamma(N,\frac{N t_m}{P_i+\sigma_w^2+ J_l})}{\Gamma(N)}\bigg) \bigg] \pi_{i,l}^{P,J} \geq U, \nonumber \\ &\quad \quad m = 1,\dots,M, \nonumber \\
& \quad\sum_{i=1}^{I} \sum_{l=1}^L \pi_{i,l}^{P,J}  = 1, \quad 0 \leq \pi_{i,l}^{P,J} \leq 1, 
\end{align}
while a Nash equilibrium mixed strategy for Willie can be found by solving the linear program:
\begin{align}
\label{jammer_Willie_LP}
&\min_{\{\pi_m^t\}, U} \,  U  \nonumber\\
&\textnormal{s.t. } \sum_{m=1}^M \bigg[ \overline{R} \bigg( \frac{P_i}{\sigma_b^2 + \alpha^2 J_l} \bigg)  \nonumber  \nonumber \\ & \quad + \beta\bigg( \frac{\Gamma(N,\frac{N t_m}{\sigma_w^2+ J_l})}{\Gamma(N)}  + 1 - \frac{\Gamma(N,\frac{N t_m}{P_i +\sigma_w^2 + J_l})}{\Gamma(N)}\bigg) \bigg] \pi_m^t \leq U, \nonumber \\ &   \quad\quad i = 1,\dots,I, \, l=1,\dots,L,  \nonumber \\
& \quad \sum_{m=1}^M \pi_m^t = 1, \quad 0 \leq \pi_m^t \leq 1.
\end{align}

Similar uniqueness properties of the Nash equilibria as discussed at the end of Section \ref{game_theoretic_sec} will also hold here. 

\begin{remark}
The linear program (\ref{Alice_jammer_LP}) is not quite in standard form, as the joint distribution $\pi^{P,J}$ is more conveniently viewed as a matrix than a vector. It can however be put into standard form by vectorizing $\pi^{P,J}$. For instance, let an index $y$ range from $1$ to $I \times L$, and consider the mappings 
\begin{align}
\label{i_l_mapping}
i(y) & \triangleq \left\{\begin{array}{ll}  I, & \textnormal{if } y \textnormal{ mod } I = 0 \\ y \textnormal{ mod } I, & \textnormal{otherwise} \end{array} \right. \nonumber \\
l(y) & \triangleq \left\lceil \frac{y}{I} \right\rceil,
\end{align}
where $\lceil . \rceil$ is the ceiling operator. 
Then the linear program (\ref{Alice_jammer_LP}) can be rewritten as:
\begin{align*}
&\max_{\{\pi_{y}\}, U} \,  U \nonumber \\
&\textnormal{s.t. }  \sum_{y=1}^{I L} \bigg[ \overline{R} \bigg( \frac{P_i}{\sigma_b^2 + \alpha^2 J_{l(y)}} \bigg) \nonumber \\ &\quad + \beta'\bigg( \frac{\Gamma(N,\frac{N t_m}{\sigma_w^2+ J_{l(y)}})}{\Gamma(N)}  + 1 - \frac{\Gamma(N,\frac{N t_m}{P_{i(y)}+\sigma_w^2+ J_{l(y)}})}{\Gamma(N)}\bigg) \bigg] \pi_y \geq U, \nonumber \\ &\quad\quad m = 1,\dots,M, \nonumber \\
& \quad \sum_{y=1}^{I L}  \pi_y  = 1, \quad 0 \leq  \pi_y \leq 1, 
\end{align*}
where $i(y)$ and $l(y)$ are replaced by the mappings (\ref{i_l_mapping}).
\end{remark}

\section{Numerical studies}
\label{numerical_sec}

\subsection{Plots of probability distributions}
We first show some plots of the Nash equilibrium mixed strategies / probability distributions. For the case with no jammer (Section \ref{game_theoretic_sec}), we use the following parameters: $\sigma_b^2 = 0 \textnormal{ dB}$, $\sigma_w^2 = 0  \textnormal{ dB}$, $\delta = 0.1$, $N=200$, $\beta = 1.6$. The transmit powers range from 0.01 mW to 1 mW in steps of 0.01 mW, and the detection thresholds are discretized from 0 to 3 in steps of 0.01. When solving the linear programs (\ref{lin_prog_Alice})-(\ref{lin_prog_Willie}), we omit values which give a negative rate in the expression (\ref{finite_blocklength_rate_approx}).\footnote{Using the above parameters, it turns out that we  omit the transmit power of 0.01 mW.}
Fig. \ref{P_dist_tx} shows the transmit power distribution and Fig. \ref{t_dist_tx} shows the threshold distribution. 
\begin{figure}[t!]
\centering 
\includegraphics[scale=0.5]{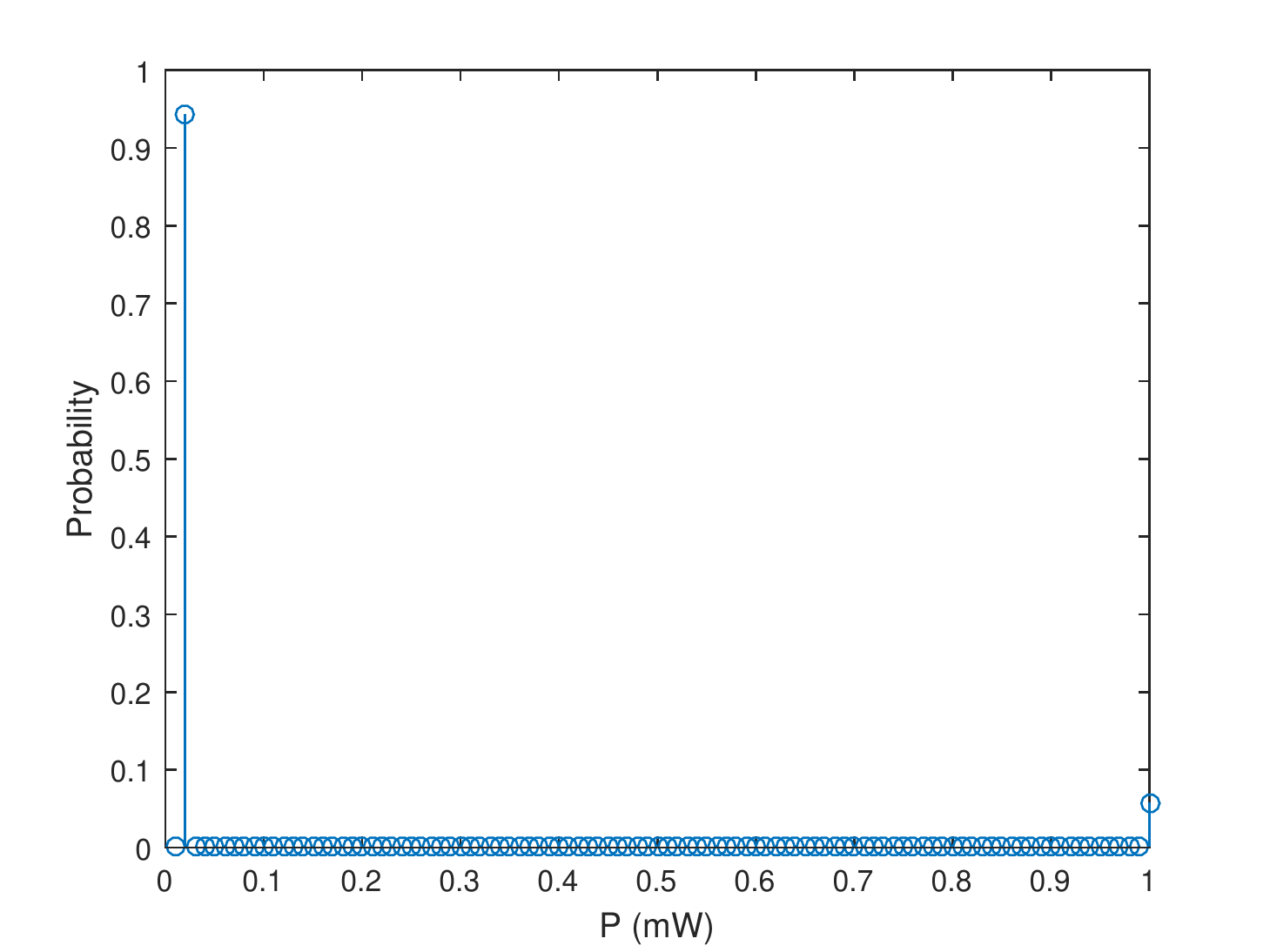} 
\caption{Transmit power distribution}
\label{P_dist_tx}
\end{figure} 
\begin{figure}[t!]
\centering 
\includegraphics[scale=0.5]{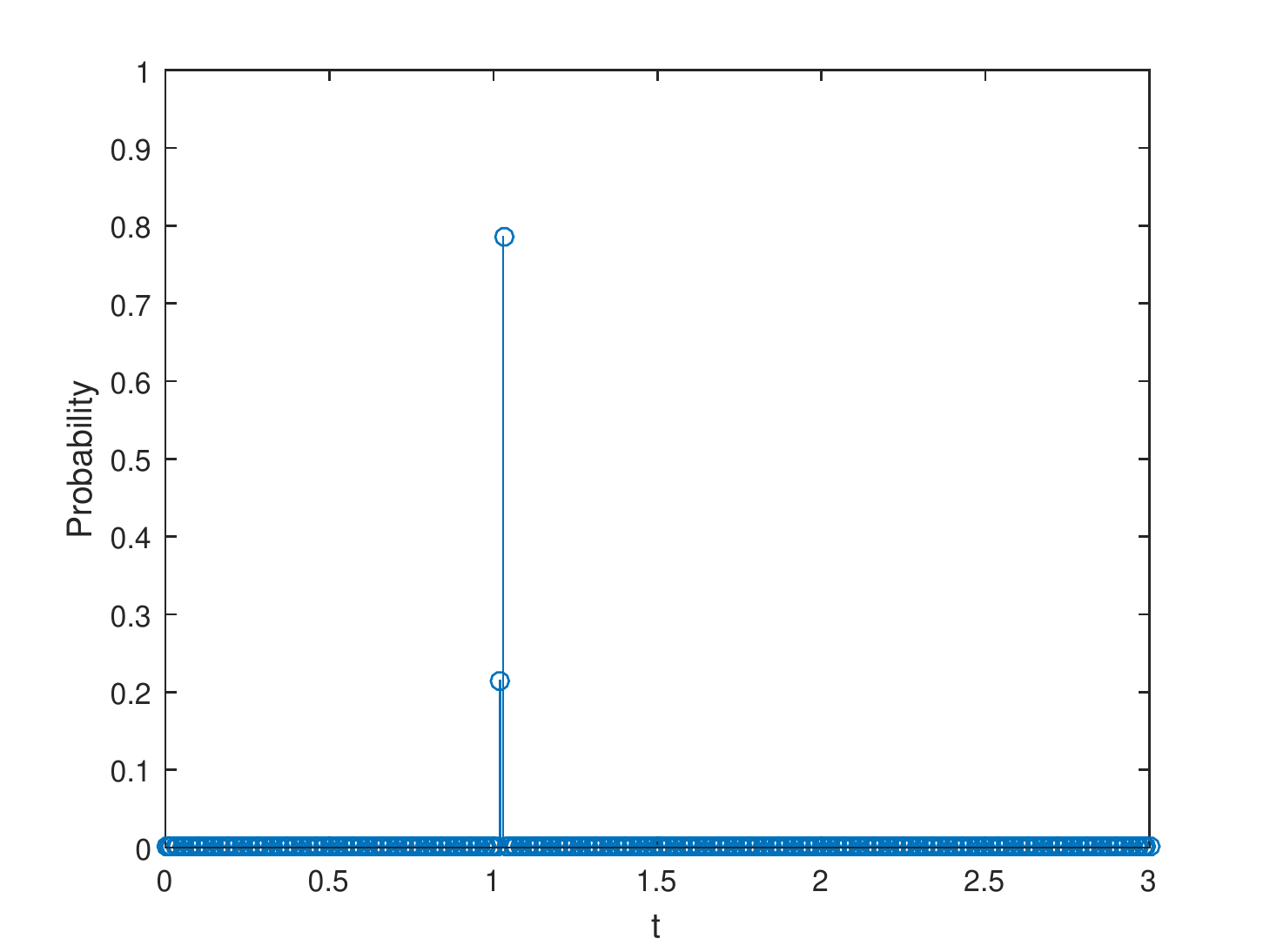} 
\caption{Detection threshold distribution}
\label{t_dist_tx}
\end{figure} 
The transmit powers here are concentrated on two values, randomizing between the lowest (0.02 mW) and highest (1 mW) power levels. The detection thresholds of Willie are randomized between the two neighbouring values 1.02 and 1.03.

In the case with a  jammer (Section \ref{jammer_power_sec}), we use the following parameters:  $\sigma_b^2 = 0 \textnormal{ dB}$, $\sigma_w^2 = 0 \textnormal{ dB}$, $\delta = 0.1$,  $\alpha = 1$, $N=200$, $\beta' = 1.5$. The transmit powers range from 0.01 mW to 1 mW in steps of 0.01 mW, the jamming powers range from 0 mW to 1 mW in steps of 0.01 mW, and the detection thresholds are discretized from 0 to 3 in steps of 0.01. When solving the linear programs (\ref{Alice_jammer_LP})-(\ref{jammer_Willie_LP}), we again omit values which give a negative rate in (\ref{finite_blocklength_rate_approx}). Fig. \ref{J_dist_jammer} shows the joint transmit and jamming power distribution and Fig. \ref{t_dist_jammer} shows the threshold distribution. 
\begin{figure}[t!]
\centering 
\includegraphics[scale=0.5]{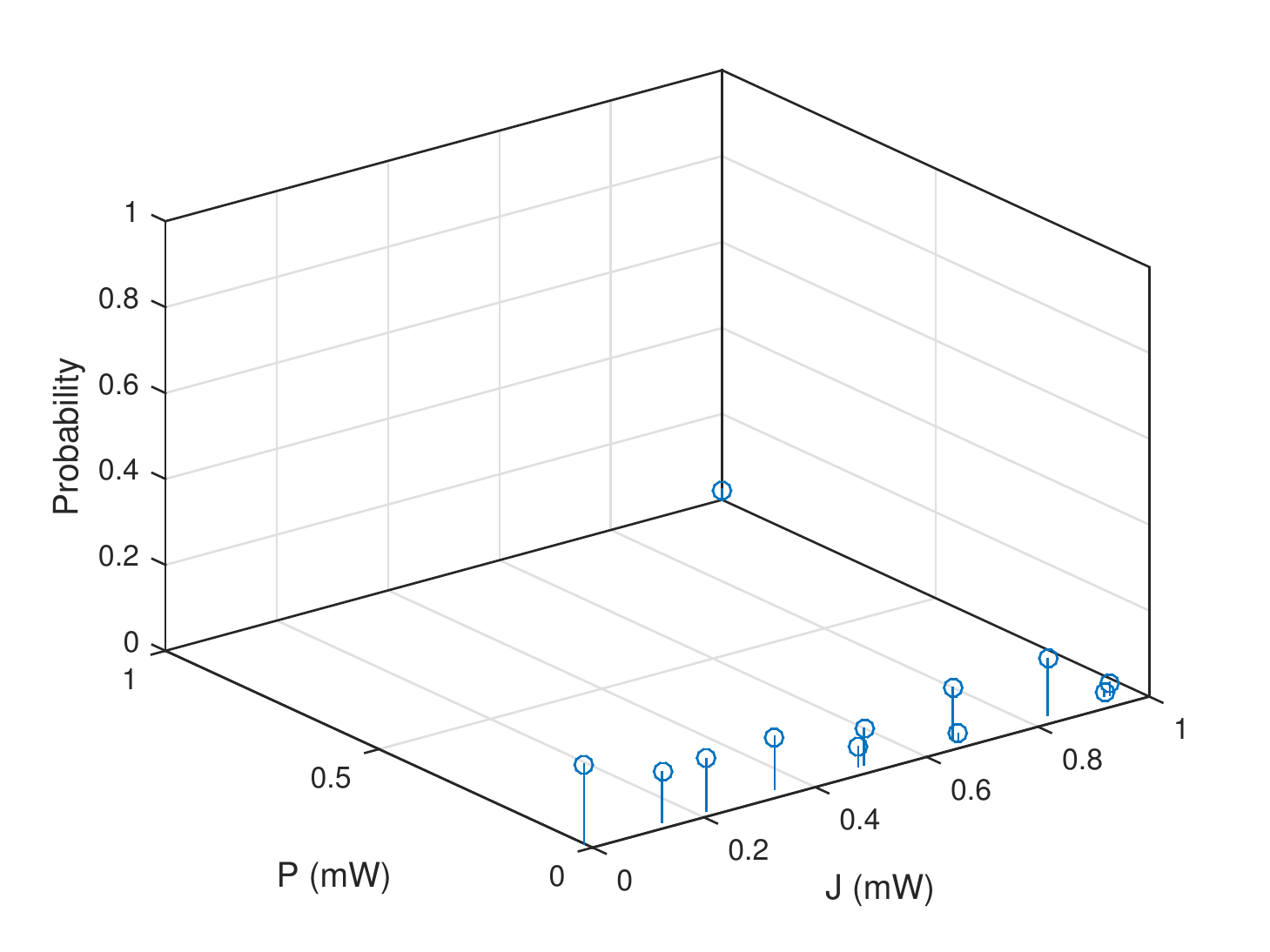} 
\caption{With jammer: Joint transmit and jamming power distribution}
\label{J_dist_jammer}
\end{figure} 
\begin{figure}[t!]
\centering 
\includegraphics[scale=0.5]{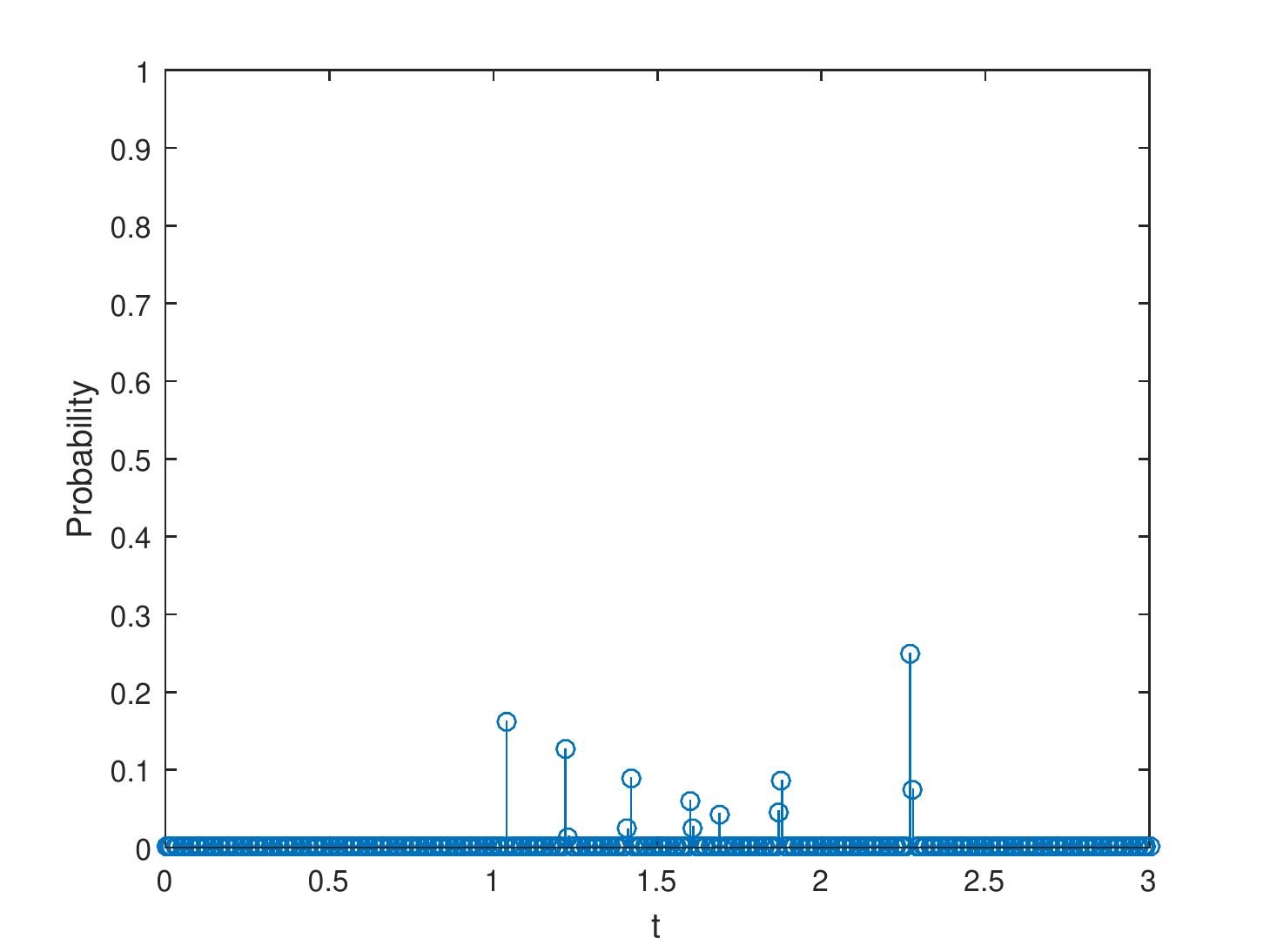} 
\caption{With jammer: Detection threshold distribution}
\label{t_dist_jammer}
\end{figure} 
The transmit-jamming powers and detection thresholds are now concentrated on multiple values. 

\subsection{Trade-off between rate and detection error probabilities}
Next we look at the trade-off between the expected coding rate per channel use and $\mathbb{P}_{FA}+\mathbb{P}_{M}$, by finding Nash equilibria for different values of $\beta$ and $\beta'$. 
In the case of the transmitter varying its transmit power, we use the following parameters: $\sigma_b^2 = 0 \textnormal{ dB}$, $\sigma_w^2 = 0 \textnormal{ dB}$, $\delta = 0.1$. 
In the case with a  jammer, we additionally set $\alpha = 1$. Fig. \ref{rate_error_plot} shows plots for various block lengths $N$. We see that in each case, the use of a jammer gives improvements in expected rate for the same covertness requirement.

Interestingly, for larger values of $\mathbb{P}_{FA}+\mathbb{P}_{M}$, when there is no jammer, the performance is not monotonic with $N$, but seems to be worse for both small and large values of $N$. For small $N$, this could be due to the finite blocklength correction in the second term of (\ref{finite_blocklength_rate_approx}), while the poorer performance for large $N$ is due to the fact that Willie can achieve better detection when he can collect more observations, and is consistent with the result from \cite{BashGoeckelTowsley} that the number of bits per channel use is $O(\sqrt{N}/N) = O(1/\sqrt{N})$ as $N \rightarrow \infty$.

On the other hand, when using a jammer, the performance appears to improve with $N$, though the improvement is slight when $N$ is large. The performance not deteriorating for large $N$ is now consistent with the result of \cite{SobersBash}, that when using a jammer the number of bits per channel use is $O(N/N) = O(1)$ as $N \rightarrow \infty$. 
\begin{figure}[t!]
\centering 
\includegraphics[scale=0.6]{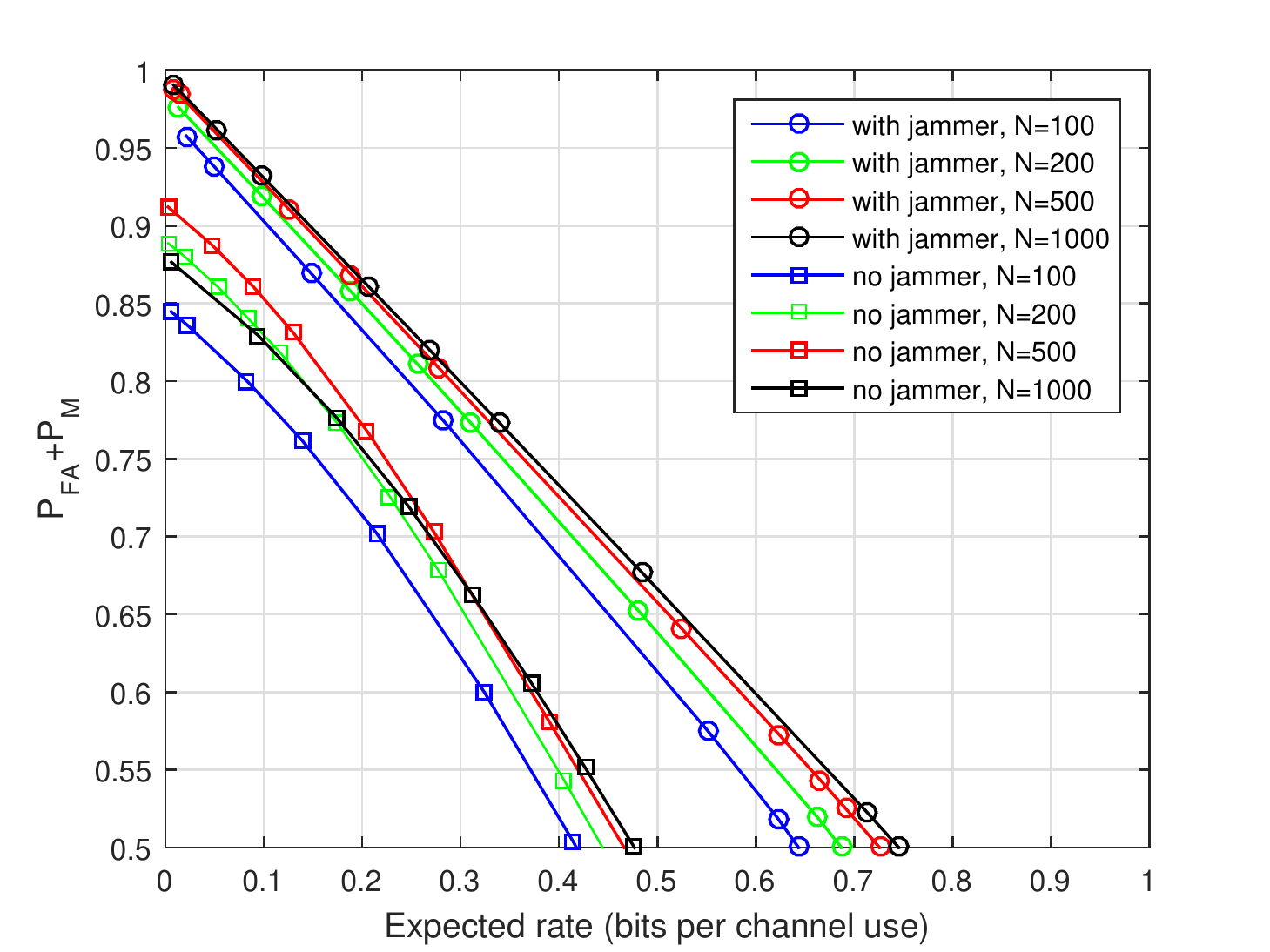} 
\caption{Expected rate per channel use vs. $\mathbb{P}_{FA}+\mathbb{P}_{M}$}
\label{rate_error_plot}
\end{figure} 

\subsection{Comparison with uniformly distributed and constant powers}
We will compare our approach with the case of uniformly distributed transmission powers that was proposed in \cite{Yan_delay_intolerant}. We consider the case $N=200$\footnote{Similar qualitative behaviour will also be observed for other values of $N$.} with the parameters $\sigma_b^2 = 0 \textnormal{ dB}$, $\sigma_w^2 = 0 \textnormal{ dB}$, $\delta = 0.1$. Fig. \ref{rate_error_plot_uniform} plots the trade-off between the expected rate per channel use and $\mathbb{P}_{FA}+\mathbb{P}_{M}$ for 1) our game-theoretic approach, 2) uniformly distributed powers, 3) constant powers. Also plotted is the performance of the game-theoretic approach with additional jammer. For uniformly distributed powers, we consider powers uniformly distributed among $(0.02 \textnormal{ mW}, 0.03 \textnormal{ mW}, \dots, 0.01 k \textnormal{ mW})$  for different values of $k \in \mathbb{N}$, in each case searching for and using  the detection threshold $t$ in $(0, 0.01, \dots, 3)$ which minimizes $\mathbb{P}_{FA}+\mathbb{P}_{M}$. For constant powers, we consider different constant transmission powers $0.02 \textnormal{ mW}, 0.03 \textnormal{ mW}, \dots$, and use in each case the detection threshold $t$ which minimizes $\mathbb{P}_{FA}+\mathbb{P}_{M}$. We see that for very strict covertness requirements (larger $\mathbb{P}_{FA}+\mathbb{P}_{M}$) all three approaches will give similar performance, but when the covertness requirement is less strict (smaller $\mathbb{P}_{FA}+\mathbb{P}_{M}$) our game-theoretic approach can achieve significantly higher rates. 
\begin{figure}[t!]
\centering 
\includegraphics[scale=0.6]{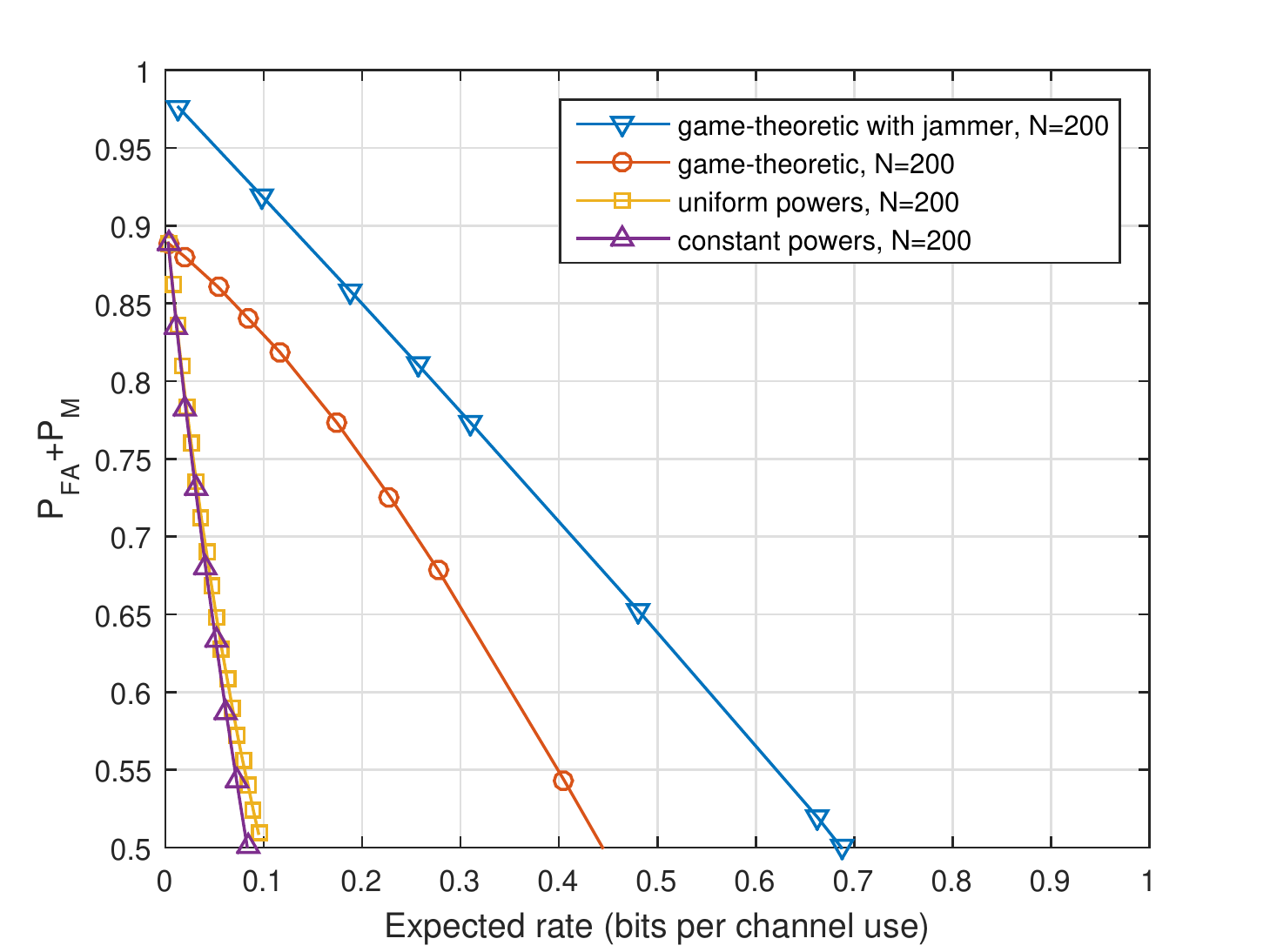} 
\caption{Expected rate per channel use vs. $\mathbb{P}_{FA}+\mathbb{P}_{M}$}
\label{rate_error_plot_uniform}
\end{figure} 

\section{Conclusion}
We have studied a game-theoretic approach to the finite blocklength covert communications problem, where Alice can randomly vary her transmit power and Willie can randomly vary his detection threshold. For less covert requirements, our game theoretic approach can achieve significantly higher coding rates than uniformly distributed transmit powers. An alternative scheme using a jammer has also been considered, with the formulation of a game between the jammer and Willie. We have shown that further performance gains can be achieved by the use of a jammer.

\bibliography{IEEEabrv,covert_estimation}
\bibliographystyle{IEEEtran}

\end{document}